\documentclass[hidelinks,onefignum,onetabnum]{siamart251216}

\usepackage{lipsum}
\usepackage{amsfonts}
\usepackage{graphicx}
\usepackage{xcolor}
\usepackage{epstopdf}
\usepackage{algorithmic}
\ifpdf

\usepackage{amsmath}
\usepackage{tikz}
\usepackage{mathdots}
\usepackage{yhmath}
\usepackage{cancel}
\usepackage{color}
\usepackage{siunitx}
\usepackage{array}
\usepackage{multirow}
\usepackage{amssymb}
\usepackage{gensymb}
\usepackage{tabularx}
\usepackage{extarrows}
\usepackage{booktabs}
\usetikzlibrary{fadings}
\usetikzlibrary{patterns}
\usetikzlibrary{shadows.blur}
\usetikzlibrary{shapes}
\usetikzlibrary{calc,positioning}

  \DeclareGraphicsExtensions{.eps,.pdf,.png,.jpg}
\else
  \DeclareGraphicsExtensions{.eps}
\fi

\newsiamremark{remark}{Remark}
\newsiamremark{hypothesis}{Hypothesis}
\crefname{hypothesis}{Hypothesis}{Hypotheses}
\newsiamthm{claim}{Claim}
\newsiamremark{fact}{Fact}
\crefname{fact}{Fact}{Facts}

\headers{T-Robinson Spaces: Structure and Recognition}{P. Asenjo, S. Cavero, M. Soto-Gomez and C. Thraves Caro}

\title{T-Robinson Spaces: Structure, Recognition, and Applications to Real Data}

\author{Patricio Asenjo\thanks{Departamento de Ingeniería Informática y Ciencias de la Computación, Facultad de Ingeniería, Universidad de Concepción
  (\email{pasenjo2018@udec.cl}).}
\and Sergio Cavero\thanks{Universidad Rey Juan Carlos, 
  (\email{sergio.cavero@urjc.es}).}
\and Mauricio Soto-Gomez\thanks{AnacletoLAB - Computational Biology and Bioinformatics Lab and Center
for Complexity and Biosystems, Università degli Studi di Milano (\email{mauricio.soto@unimi.it}).}
\and Christopher Thraves Caro\thanks{Departamento de Ingeniería Matemática, Facultad de Ciencias Físicas y Matemáticas, Universidad de Concepción
  (\email{cthraves@udec.cl}).}
  }

\usepackage{amsopn}

\usepackage{pgfplots}
\pgfplotsset{compat=1.18}
\usepackage{xspace}
\usepackage{adjustbox}
\usepackage{hyperref}

\newcommand{\Trob}{$T$-Robinson\xspace}
\newcommand{\Bxy}{\ensuremath{B_{xy}}\xspace}

\ifpdf
\hypersetup{
  pdftitle={T-Robinson Spaces: Structure, Recognition, and Applications to Real Data},
  pdfauthor={P. Asenjo, S. Cavero, M. Soto-Gomez and C. Thraves Caro}
}
\fi

\begin{document}

\maketitle

\begin{abstract}
We study \emph{\Trob spaces}, a tree-based generalization of Robinson spaces in which every path of a compatible tree induces a Robinson subspace. This framework extends the classical notion of Robinsonian representations from linear orderings to tree structures, allowing the modeling of hierarchical and branching data. We establish a complete combinatorial characterization of \Trob  spaces by proving their equivalence with several graph- and hypergraph-theoretic properties. In particular, we show that a dissimilarity space is \Trob if and only if all its level graphs are dually chordal with a common compatible tree. Combined with the characterization of hypertrees established by Brucker~\cite{brucker2005hypertrees}, this yields the equivalent characterization in terms of the associated cluster, ball, and 2-ball hypergraphs being hypertrees. Building upon these structural results, we develop a  recognition algorithm with complexity \(O(K n^{2})\), where \(K\) denotes the number of minimum spanning trees of the dissimilarity space, improving upon existing hypertree-based approaches whenever \(K\) remains moderate. We further introduce a quantitative measure of \Trob structure that evaluates the extent to which an arbitrary dissimilarity space admits a tree-like representation. Finally, we discuss applications to real-world datasets, illustrating how \Trob spaces provide an interpretable framework for analyzing and organizing relational data.
\end{abstract}



\section{Introduction}

The exponential growth of relational data has generated an increasing demand for mathematical tools capable of representing, organizing, and analyzing complex dissimilarity structures. Pairwise dissimilarity information naturally arises in a wide variety of domains, including bioinformatics, phylogenetics, pangenome analysis, network science, clustering, and information retrieval \cite{costa2020dissimilarity,liu2006clustering,penner2011sequence}. In many of these applications, the available information consists not of coordinates in a geometric space, but rather of pairwise comparisons or dissimilarity measurements between objects. Understanding how such data can be represented in a structured and interpretable way is, therefore, a fundamental problem at the intersection of mathematics, computer science, and data analysis.

One of the most influential frameworks for representing dissimilarity data is provided by Robinson spaces. Originating in the context of seriation, Robinson spaces capture the principle that similar objects should appear close together in a suitable ordering, while dissimilar objects should appear farther apart \cite{liiv2010seriation}. Formally, a Robinson ordering arranges the objects along a line so that dissimilarities increase monotonically when moving away from the diagonal of the corresponding dissimilarity matrix. This simple idea has led to a rich mathematical theory connecting graph representations, combinatorial optimization, metric geometry, and algorithm design \cite{brucker2005hypertrees,caraux2005permutmatrix,chen2004matrix}.

Despite their success, Robinson spaces suffer from an inherent limitation: they are fundamentally one-dimensional. Many datasets exhibit a hierarchical or branching structure that cannot be adequately represented by a linear ordering \cite{henderson2013using}. Phylogenetic relationships, taxonomic classifications, evolutionary histories, and numerous network datasets often exhibit a tree-like structure \cite{abu2016metric}, for which a single linear arrangement fails to capture the underlying geometry. This observation naturally raises the question of whether the Robinson paradigm can be extended from linear representations to tree representations.

In this work, we study \emph{\Trob spaces}, a class of dissimilarity spaces admitting a tree-based Robinson representation.
This notion generalizes classical Robinson spaces by replacing the underlying line with a tree, while preserving the fundamental principle that dissimilarities should increase as one travels along the structure of the representation. 
%
%
By treating  a dissimilarity matrix as the weighted adjacency matrix of a complete graph,
we provide a structural characterization of \Trob spaces, establishing a precise connection between these spaces and several well-studied combinatorial structures, including dually chordal graphs and hypertrees. These connections reveal a rich underlying combinatorial framework and allow us to derive necessary and sufficient conditions for the existence of a \Trob representation.

Building upon this characterization, we develop a recognition algorithm for \Trob spaces. Existing approaches rely on hypertree recognition, leading to algorithms with complexity \(O(n^5)\) \cite{brucker2005hypertrees}. By exploiting the structural properties uncovered in our characterization, we obtain a recognition algorithm with complexity \(O(Kn^2)\), where \(K\) denotes the number of distinct minimum spanning trees of the input dissimilarity space. This significantly reduces the complexity of the recognition problem whenever the number of minimum spanning trees remains moderate, while simultaneously providing a deeper understanding of the relationship between \Trob spaces and their associated spanning-tree structures.

Beyond exact recognition, we introduce a quantitative measure of \Trob structure that evaluates how closely an arbitrary dissimilarity space resembles a \Trob space. This measure allows us to move beyond the binary question of recognition and to investigate the extent to which tree-Robinson structure is present in real-world datasets. We apply this methodology to several real-data collections, providing the first empirical study of the \Trob structure in practical applications.

\section{Definitions}
Let $X$ be a finite set. A \emph{dissimilarity} on $X$ is a symmetric function
\(
D:X\times X\to \mathbb{R}_{\ge 0}
\)
such that $D(x,y)=0$ if $x=y$. The pair $(X,D)$ is called a
\emph{dissimilarity space}. Given $Y\subseteq X$, the restriction
$(Y,D|_{Y\times Y})$ is the \emph{subspace} of $(X,D)$ induced by $Y$. By abuse of notation, we will often denote this subspace simply by $(Y,D)$.

A dissimilarity space $(X,D)$ is \emph{Robinson} if there exists a total order
$<$ on $X$ such that, for every triple $x<y<z$,
\[
D(x,z)\ge \max\{D(x,y),D(y,z)\}.
\]
Such an order is called a \emph{compatible order}.

The objective of this study is to generalize the concept of Robinson spaces from linear representations to tree representations. This motivates the following definition.
\begin{definition}[\Trob]\label{def:trob}
Let $(X,D)$ be a dissimilarity space.  
We say that $(X,D)$ is \emph{\Trob} if there is a tree $T$ with vertex set $X$, such that for every path $P$ of $T$, the subspace
\(
(V(P),D)
\)
induced by $P$ is Robinson with respect to the order induced by \(P\). In this case, $T$
is called a \emph{compatible tree} of $(X,D)$.
\end{definition}

%

It is worth noticing that every Robinson space is \Trob. Indeed, for every compatible ordering of a Robinson space, the path whose vertices are ordered according to the (Robinson) compatible order is a compatible tree. Conversely, not every \Trob space is necessarily Robinson. Figure~\ref{fig:Trobinson_example} depicts an example of a \Trob space which is not Robinson. 

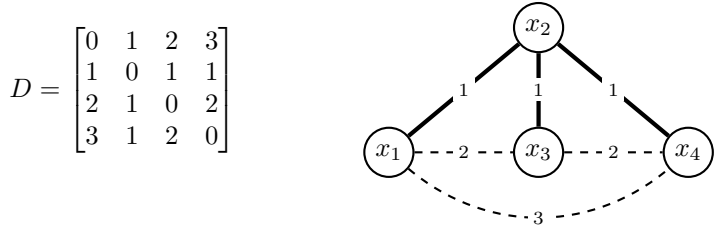
\begin{figure}
    \centering
    \begin{tikzpicture}[
        scale=1.1,
        node/.style={circle, draw=black, thick, fill=white, minimum size=6.5mm, inner sep=0pt},
        tree edge/.style={ultra thick, black},
        valid path/.style={thick, dashed, draw=black},
        invalid path/.style={thick, dashed, draw=red!80!black},
        lbl/.style={midway, fill=white, inner sep=2pt, font=\scriptsize, text=black}
    ]
    \begin{scope}
        \node (D) {$D = \begin{bmatrix} 
0 & 1 & 2 & 3 \\ 
1 & 0 & 1 & 1 \\ 
2 & 1 & 0 & 2 \\ 
3 & 1 & 2 & 0 
\end{bmatrix}
$};
    \end{scope}
        
        \begin{scope}[shift={(5,-.75)}]
            \node[node] (x2) at (0, 1.5) {$x_2$};
            \node[node] (x1) at (-1.8, 0) {$x_1$};
            \node[node] (x3) at (0, 0) {$x_3$};
            \node[node] (x4) at (1.8, 0) {$x_4$};
            
            \draw[tree edge] (x2) to node[lbl] {$1$}  (x1);
            \draw[tree edge] (x2) to node[lbl] {$1$} (x3);
            \draw[tree edge] (x2) to node[lbl] {$1$} (x4);
            
            \draw[valid path, bend right=0] (x1) to node[lbl] {$2$} (x3);
            \draw[valid path, bend right=40] (x1) to node[lbl] {$3$} (x4);
            \draw[valid path, bend right=0] (x3) to node[lbl] {$2$} (x4);
        \end{scope}
    \end{tikzpicture}
    \caption{Example of a \Trob space. The matrix on the left gives the dissimilarity values, while the graph on the right displays a compatible tree, whose edges are drawn as solid lines. Notice that the space is T-Robinson but not Robinson.}
    \label{fig:Trobinson_example}
\end{figure}

To establish our structural characterization of \Trob spaces, we will use several graph- and hypergraph-theoretic notions, which we introduce below.
To this end, we interpret the dissimilarity values as weights on an undirected complete graph defined over the set of elements $X$. 
By an abuse of notation, we refer to both the dissimilarity space and its induced weighted complete graph as $(X, D)$. 
Furthermore, in this graph, we denote the edge $\{x, y\}$ between elements $x$ and $y$ simply as $xy$.

For $\alpha>0$, the $\alpha$-\emph{level graph} of $(X,D)$ is the graph
\(
G_\alpha=(X,E_\alpha)
\),
where $xy$ belongs to $E_\alpha$ if $D(x,y)\le \alpha$. A maximal clique of an $\alpha$-level graph is called a \emph{cluster} of $(X,D)$.
Given $x\in X$ and $r>0$, the \emph{ball} centered at $x$ with radius $r$ is
\[
B_r(x)=\{z\in X:\ D(x,z)\le r\}.
\]
For $x,y\in X$, the \emph{2-ball} induced by $x$ and $y$ is
\[
\Bxy
=
\{z\in X:\ \max\{D(x,z),D(y,z)\}\le D(x,y)\}.
\]

Associated with a dissimilarity space $(X,D)$, we consider the following
hypergraphs with vertex set $X$:
     the \emph{cluster hypergraph}, whose hyperedges are all clusters of
    all level graphs of $(X,D)$;
     the \emph{ball hypergraph}, whose hyperedges are all balls of
    $(X,D)$;
     the \emph{2-ball hypergraph}, whose hyperedges are all 2-balls of
    $(X,D)$.

A hypergraph $\mathcal H=(X,\mathcal E)$ is called a \emph{hypertree} if there
exists a tree $T$ on $X$ such that every hyperedge of $\mathcal H$
induces a subtree of $T$. Such a tree is called an \emph{underlying tree} of
$\mathcal H$.

A graph $G=(X,E)$ is called \emph{dually chordal} if there exists a spanning
tree $T$ of $G$ such that every closed neighborhood of $G$ induces a subtree
of $T$. Equivalently, Brandstat et al. in \cite{brandstadt1998dually} show that $G$ is dually chordal if there exists a
spanning tree $T$ of $G$ such that every maximal clique of $G$ induces a
subtree of $T$. Such a tree is called a \emph{dually chordal compatible tree} of $G$.

\section{State of the art and our results}
Robinson's 1951 introduction of an ordering method for archaeological deposits \cite{robinson1951method} marked a foundational step in understanding structured data. Building on this work, Kendall  \cite{kendall1969incidence} gave the first definition and characterization of Robinson matrices, which are dissimilarity matrices of Robinson spaces written according to a compatible ordering. 

Since then, several recognition methods for Robinson spaces have been studied. Among them, we highlight two optimal polynomial-time algorithms. The first one, due to Pr\'ea and Fortin~\cite{prea2014optimal}, uses PQ-trees, yielding an optimal algorithm with complexity \(O(n^2)\). The second one, due to Carmona et al.~\cite{carmona2024modules}, is also an optimal \(O(n^2)\) algorithm and is based on the study of \emph{mmodules} in Robinson spaces.
The recognition problem becomes substantially more difficult when the dissimilarity matrix is only partially specified; Aracena and Thraves~\cite{aracena2021recognition} showed that recognizing Robinson spaces with missing data is NP-hard.

Approximation problems in Robinson spaces are important because real data often contains missing or flawed information. Barthélemy and Brucker  \cite{barthelemy2001np} studied the problem of approximating a dissimilarity space by a Robinson space through a minimization problem under an \(L^p\) distance, for \(p\in[0,+\infty)\), and proved that this optimization problem is NP-hard. Later, Chepoi et al. \cite{chepoi2009seriation} extended this hardness result to the case \(p=+\infty\).

The notion of hypergraph was introduced by Berge~\cite{berge1985graphs} as a generalization of graphs, allowing one to represent relations that are not necessarily pairwise. Hypertrees constitute a particularly relevant class of hypergraphs for the present work. Flament~\cite{flament1978hypergraphes} established a characterization of hypertrees: a hypergraph \(H=(X,\mathcal E)\) is a hypertree if and only if it satisfies the following two properties. First, it satisfies the \emph{Helly} property, namely, for every subfamily \(\mathcal E'\subseteq \mathcal E\) such that any two distinct hyperedges 
in \(\mathcal E'\) intersect, one has
\(
\cap_{e\in \mathcal E'} e \neq \emptyset
\).
Second, its line graph \(L(H)\) is chordal, where \(V(L(H))=\mathcal E\) and two vertices of \(L(H)\) are adjacent if and only if the corresponding hyperedges intersect in \(H\).

The motivation for studying \Trob spaces comes from Brucker's work~\cite{brucker2005hypertrees} on a particular class of dissimilarities: tree quasi-ultrametrics. Diatta and Fichet~\cite{diatta1994apresjan} define a quasi-ultrametric \(D:X\times X\to \mathbb R_{\ge 0}\) as a function satisfying, for all \(x,y,z\in X\), if
 \(\max\{D(x,z),D(y,z)\}\le D(x,y)\)
then for all  $t$ in  $X$, 
\[D(t,z)\le \max\{D(x,t),D(y,t),D(x,y)\}.\]
This condition is known as the four-point condition. A tree quasi-ultrametric is then a quasi-ultrametric whose cluster hypergraph is a hypertree.  Brucker in \cite{brucker2005hypertrees} proved that hypertrees admit a characterization based on an ordering of their vertices, which also provides a way to identify one of their underlying trees. Moreover, Brucker showed that, for any dissimilarity space, the cluster hypergraph, ball hypergraph, and 2-ball hypergraph are equivalent in the sense that if one of them is a hypertree, then the other two are hypertrees as well. Finally, Brucker provided an algorithm for recognizing whether a hypergraph is a hypertree in time $O(n^3|\mathcal E|)$, where $n=|X|$ is the number of vertices of the hypergraph. 

We show that a dissimilarity space is T-Robinson if and only if its 2-ball hypergraph is a hypertree. We further prove that every compatible tree of a \Trob space is necessarily a minimum spanning tree (MST) of the underlying dissimilarity space, a property that leads to an efficient recognition algorithm. Recently, Préa~\cite{prea2026tree} investigated a stronger notion of \Trob spaces, in which \emph{every} minimum spanning tree is compatible, and proposed an $O(n^3)$-time recognition algorithm for this more restrictive class.

Circular seriation extends the classical seriation problem by replacing linear orders with circular ones. In this direction, Armstrong et al.~\cite{armstrong2021optimal} introduced the notion of \emph{circular Robinson matrices} as the natural analogue of Robinson matrices for circular orders, obtained an optimal \(O(n^2)\) recognition algorithm for strict circular Robinson spaces, and established statistical guarantees on the recovery of the underlying circular order. More recently, Carmona et al.~\cite{carmona2023simple} significantly simplified this recognition problem by proposing an \(O(n\log n)\) algorithm for computing a compatible circular order in the strict setting, complemented with an \(O(n^2)\) verification algorithm. They also proved the equivalence between circular Robinson dissimilarities and pre-circular Robinson dissimilarities, thereby providing a structural characterization of this class.

\subsection*{Our contributions and structure of the document}

 In Section~\ref{sec:characterization}, we develop the structural theory of \Trob spaces. As an intermediate result, we establish in Theorem~\ref{lem:dually-chordal-t-clique} a new characterization of dually chordal graphs. Besides being of independent interest, this characterization plays a central role in proving our main structural result, namely the following characterization of \Trob spaces in terms of dually chordal graphs and hypertrees.

\begin{theorem}\label{thm:characterization}
Let $(X,D)$ be a dissimilarity space and let $T$ be a tree on the vertex set $X$. The following statements are equivalent:
\begin{enumerate}
 \item $(X,D)$ is a \Trob space with compatible tree $T$.
 \item For all $\alpha>0$, the $\alpha$-level graph $G_\alpha$ is dually chordal with compatible tree $T$.
 \item The cluster hypergraph of $(X,D)$ is a hypertree with underlying tree $T$.
 \item The ball hypergraph of $(X,D)$ is a hypertree with underlying tree $T$.
 \item The $2$-ball hypergraph of $(X,D)$ is a hypertree with underlying tree $T$.
\end{enumerate}
\end{theorem}

In Section~\ref{sec:recognition}, we study the recognition problem for \Trob spaces. A key ingredient is Lemma~\ref{thm:mst}, where we prove that every compatible tree of a \Trob space must be a minimum spanning tree of the dissimilarity space. This structural property is crucial for restricting the search space of candidate compatible trees and forms the basis of our recognition algorithm. Using this result, we establish the following theorem.

\begin{theorem}
    \label{thm:rob_recognition}
    Let $(X,D)$ be a dissimilarity space on $n$ elements and let $K$ denote the number of minimum spanning trees of $(X,D)$.
    One can decide whether $(X,D)$ is a \Trob space in $O(K n^2)$ time and $O(n^2)$ space.
\end{theorem}

In Section~\ref{sec:optimization_approach}, we move beyond the recognition problem and introduce a quantitative measure of \Trob structure. We then propose a heuristic based on minimum spanning trees that aims to find spanning trees maximizing this measure. Finally, in Section~\ref{sec:experiments}, we evaluate the proposed methodology on real-world datasets and provide, to the best of our knowledge, the first empirical analysis of \Trob structure in practice. Together, these results establish new theoretical foundations for tree-based Robinson representations and provide practical tools for the analysis of relational data.

\section{Characterization}\label{sec:characterization}
In this section, we establish the combinatorial characterization of \Trob spaces, which constitutes the main structural result of this work. Our approach relies on a sequence of equivalences connecting \Trob spaces with graph- and hypergraph-theoretic structures. As a first step, we revisit the class of dually chordal graphs and provide a new characterization in terms of a spanning-tree property.

Let $G$ be a graph and let $T$ be a spanning tree of $G$. We say that $T$ satisfies the \emph{T-clique property} for $G$ if, for every edge $xy\in E(G)$, the vertices of the unique $x$--$y$ path in $T$ induce a clique in $G$.

As we shall see, this property provides a natural bridge between dually chordal graphs and \Trob spaces, and will play a central role in the development of our main characterization theorem.

\begin{theorem}\label{lem:dually-chordal-t-clique}
Let $G=(X,E)$ be a graph. Then,  $G=(X,E)$ is a dually chordal graph with compatible tree $T$ if and only if $T$ satisfies the T-clique property for $G$.
\end{theorem}

\begin{proof}
Suppose that $T$ satisfies the T-clique property for $G$. We prove
that $G$ is dually chordal with compatible tree $T$. By the characterization of dually chordal graphs in terms of maximal cliques, it is enough to show that every maximal clique of $G$ induces a subtree of $T$.

Let $C$ be a maximal clique of $G$. Assume, for a contradiction, that $T[C]$ is not connected. Then there exist vertices $x,y\in C$ such that the unique $x$--$y$ path in $T$ contains a vertex $z\notin C$. Since $C$ is a clique, we have $xy\in E$. As  $T$ satisfies the T-clique property for $G$, all vertices of the $x$--$y$ path in $T$ induce a clique in $G$. In particular, $z$ is adjacent to both $x$ and $y$.

We claim that $z$ is adjacent to every vertex of $C$. Indeed, let $w\in C$. Since $C$ is a clique, both $wx$ and $wy$ are edges of $G$. Moreover, as $z$ lies on the $x$--$y$ path in the tree $T$, it follows that $z$ lies on at least one of the paths from $w$ to $x$ or from $w$ to $y$. Applying the T-clique property to the corresponding edge, we conclude that $wz\in E$. Thus $z$ is adjacent to every vertex of $C$, and therefore $C\cup\{z\}$ is a clique of $G$, contradicting the maximality of $C$. Hence $T[C]$ is connected for every maximal clique $C$, and so $G$ is dually chordal with compatible tree $T$.

Conversely, suppose that $G$ is dually chordal, and let $T$ be a spanning tree of $G$ such that every maximal clique of $G$ induces a subtree of $T$. We prove that  $T$ satisfies the T-clique property for $G$. 

Let $xy\in E$. Choose a maximal clique $C$ of $G$ containing both $x$ and $y$. Since $T[C]$ is a subtree of $T$, the unique $x$--$y$ path in $T$ is entirely contained in $C$. As $C$ is a clique of $G$, all vertices of this path induce a clique in $G$. Since the edge $xy$ was arbitrary,  $T$ satisfies the T-clique property for $G$. 
\end{proof}

We now show a fundamental lemma that offers an alternative description of the \Trob property. This reformulation will serve as a key tool for the remainder of the section.
\begin{lemma}\label{lem:path_characterization}
    Let $(X,D)$ be a dissimilarity space. 
    The space $(X,D)$ is \Trob with compatible tree $T$ if and only if for every $x-y-$path, and for every pair of elements $v$ and $w$ in the path it holds $D(v,w)\leq D(x,y)$.
\end{lemma}
\begin{proof} 
Suppose first that $(X,D)$ is \Trob with compatible tree $T$. Assume, for a contradiction, that there exist vertices $x,y,v,w\in X$ such that $v$ and $w$ belong to the $x$--$y$ path of $T$ and \( D(v,w)>D(x,y)\). Since $(X,D)$ is \Trob, every path of $T$ induces a Robinson ordering. Without loss of generality, assume that $v$ lies on the $x$--$w$ path. Since the path is Robinson, we obtain \[ D(x,w)\ge \max\{D(x,v),D(v,w)\}, \] and therefore \( D(x,w)\ge D(v,w)\). On the other hand, because $w$ lies on the $x$--$y$ path, the Robinson property also yields \(D(x,y)\ge D(x,w)\). Combining these inequalities gives \[ D(v,w)>D(x,y)\ge D(x,w)\ge D(v,w), \] a contradiction. Hence, for every $x$--$y$ path in $T$ and every pair of vertices $v,w$ on that path, we have \( D(v,w)\le D(x,y) \). 

Conversely, suppose that $T$ is a tree on vertex set $X$ such that, for every $x$--$y$ path in $T$ and every pair of vertices $v,w$ on that path, \( D(v,w)\le D(x,y) \). Let $x,y\in X$ and let $P$ be the unique $x$--$y$ path in $T$. For any vertex $z$ of $P$, the hypothesis implies that \[ \max\{D(x,z),D(y,z)\}\le D(x,y). \] Therefore, the subspace induced by $V(P)$ is Robinson with respect to the order induced by the path. Since $x$ and $y$ were arbitrary, every path of $T$ induces a Robinson subspace. Consequently, $(X,D)$ is \Trob with compatible tree $T$. 
\end{proof}

We now present our first characterization of \Trob spaces using dually chordal graphs.

\begin{lemma}\label{lem:level_graphs_dually_chordal}
Let $(X,D)$ be a dissimilarity space and let $T$ be a spanning tree of $(X,D)$.  Then $(X,D)$ is a \Trob space with compatible tree $T$ if and only if for every $\alpha>0$, the $\alpha$-level graph $G_\alpha$ is dually chordal with compatible tree $T$.
\end{lemma}

\begin{proof}
Suppose that $(X,D)$ is \Trob with compatible tree $T$.
Let $\alpha>0$ and let $G_\alpha$ be the corresponding $\alpha$-level graph. Consider an edge $xy\in E(G_\alpha)$. Then $D(x,y)\le \alpha$.

Let $u$ and $v$ be any two vertices on the $x$--$y$ path of $T$.
By Lemma~\ref{lem:path_characterization}, \(D(u,v)\le D(x,y)\le \alpha\).
Hence $uv\in E(G_\alpha)$. Since this holds for every pair of vertices on the $x$--$y$ path, the vertices of this path induce a clique in $G_\alpha$. Therefore, $T$ satisfies the T-clique property for $G_\alpha$, and hence $G_\alpha$ is dually chordal.

Suppose that every level graph $G_\alpha$ is dually chordal with a compatible tree $T$. By the characterization of dually chordal graphs due to Brandstädt et al. \cite{brandstadt1998dually}, every closed neighborhood in $G_\alpha$ induces a subtree of $T$.

Assume, for a contradiction, that $(X,D)$ is not \Trob with compatible tree $T$. By Lemma~\ref{lem:path_characterization}, there exist vertices $x,y,z\in X$ such that $y$ lies on the $x$--$z$ path of $T$ and \(D(x,y)>D(x,z)\).

Set
\(\alpha^{*}=D(x,z)\). Then $xz\in E(G_{\alpha^{*}})$, and therefore
\(
z\in N_{G_{\alpha^{*}}}[x]
\).
On the other hand,
\(
D(x,y)>\alpha^{*}
\),
so $y\notin N_{G_{\alpha^{*}}}[x]$.

Since $y$ lies on the $x$--$z$ path of $T$, the unique path in $T$ joining two vertices of $N_{G_{\alpha^{*}}}[x]$ contains a vertex outside this neighborhood. Consequently, $N_{G_{\alpha^{*}}}[x]$ does not induce a subtree of $T$, contradicting the dually chordal property.

Therefore, $(X,D)$ is \Trob with compatible tree $T$. 
\end{proof}

We now turn to the hypergraph perspective and present a characterization of \Trob spaces in terms of the hypergraphs naturally associated with a dissimilarity space. In particular, we show that the \Trob property can be completely captured by the structure of the cluster, ball, and 2-ball hypergraphs. We begin with an instrumental lemma describing the structure of 2-balls in a \Trob space and their relationship with compatible trees. This result will serve as the key ingredient in establishing the hypergraph characterization.

\begin{lemma}\label{lem:two_balls_paths} 
Let $(X,D)$ be a \Trob space with compatible tree $T$. For every pair of vertices $x,y\in X$, every vertex of the $x$--$y$ path in $T$ belongs to the 2-ball $\Bxy$. Moreover, for every $w\in \Bxy$, both the $w$--$x$ path and the $w$--$y$ path in $T$ are contained in $\Bxy$. \end{lemma} 
\begin{proof} 
Let $x,y\in X$, and let $z$ be a vertex on the $x$--$y$ path in $T$. Since $T$ is compatible with the \Trob property, the path between $x$ and $y$ induces a Robinson subspace. Hence 
\[ 
\max\{D(x,z),D(y,z)\}\le D(x,y), 
\] 
and therefore $z\in \Bxy$. 

Now let $w\in \Bxy$. If $w$ lies on the $x$--$y$ path in $T$, then the previous argument implies that both the $w$--$x$ path and the $w$--$y$ path are contained in $\Bxy$. Suppose then that $w$ does not lie on the $x$--$y$ path. Let $c$ be the unique vertex at which the three paths from $w$ to $x$, from $w$ to $y$, and from $x$ to $y$ meet. We prove that every vertex of the $c$--$w$ path belongs to $\Bxy$. Let $h$ be a vertex on the $c$--$w$ path. Since $T$ is compatible with the \Trob property and $h$ lies on the $x$--$w$ path and on the $y$--$w$ path, we have 
\[ 
D(h,x)\le D(x,w) \quad\text{and}\quad D(h,y)\le D(y,w). 
\] 
As $w\in \Bxy$, we also have 
\[ 
D(x,w)\le D(x,y) \quad\text{and}\quad D(y,w)\le D(x,y). 
\] 
Thus \( \max\{D(h,x),D(h,y)\}\le D(x,y)\), and so $h\in \Bxy$. 

Consequently, the whole $c$--$w$ path is contained in $\Bxy$. Since the $w$--$x$ path is the union of the $w$--$c$ path and the $c$--$x$ path, and the $c$--$x$ path is contained in the $x$--$y$ path, it follows that the $w$--$x$ path is contained in $\Bxy$. The same argument applies to the $w$--$y$ path. This proves the result. 
\end{proof}

We now present our second characterization of \Trob spaces, this time from a hypergraph-theoretic perspective.

\begin{lemma}\label{lem:two_ball_hypertree}
Let $(X,D)$ be a dissimilarity space and let $T$ be a spanning tree of $(X,D)$. Then $(X,D)$ is \Trob with compatible tree $T$ if and only if the
2-ball hypergraph of $(X,D)$ is a hypertree with underlying tree $T$.
\end{lemma}

\begin{proof}
Suppose first that $(X,D)$ is \Trob with compatible tree $T$. We prove
that every 2-ball induces a subtree of $T$. Let $x,y\in X$ and consider the 2-ball \(\Bxy=\{z\in X:\max\{D(x,z),D(y,z)\}\le D(x,y)\}\).
We claim that $T[\Bxy]$ is connected. Let $u,v\in \Bxy$, and let $w$ be a vertex on the unique $u$--$v$ path in $T$. Since $T$ is a tree, $w$ lies either on the $u$--$x$ path or on the $v$--$x$ path. By \cref{lem:two_balls_paths}, both the $u$--$x$ path and the $v$--$x$ path are contained in $\Bxy$. Hence $w\in \Bxy$. Therefore, every path in $T$ joining two vertices of $\Bxy$ is contained in $\Bxy$, and so $T[\Bxy]$ is a subtree of $T$. Thus, the 2-ball hypergraph is a hypertree with underlying tree $T$.

Conversely, suppose that the 2-ball hypergraph of $(X,D)$ is a hypertree with underlying tree $T$. Let $x,y\in X$, and let $P$ be the unique $x$--$y$ path in $T$. Since every 2-ball induces a subtree of $T$, the subgraph $T[\Bxy]$ is a subtree. As $x,y\in \Bxy$, the unique $x$--$y$ path in $T$ is contained in $T[\Bxy]$. Hence every vertex $z\in V(P)$ belongs to $\Bxy$, and therefore
\[
\max\{D(x,z),D(y,z)\}\le D(x,y).
\]
Thus, the subspace induced by $P$ is Robinson with respect to the order induced by the path. Since $x$ and $y$ were arbitrary, every path of $T$ induces a Robinson subspace. Therefore, $(X,D)$ is \Trob with compatible tree $T$.
\end{proof}

We are now in a position to prove \Cref{thm:characterization}, which brings together the graph-theoretic and hypergraph-theoretic characterizations developed throughout this section.

\begin{proof}[Proof of Theorem \ref{thm:characterization}]
\Cref{lem:level_graphs_dually_chordal} establishes 
the equivalence between $1$ and $2$. While the equivalence between $1$
and $3$ follows from \Cref{lem:two_ball_hypertree}. Finally, by
Proposition~3 of \cite{brucker2005hypertrees}, the cluster hypergraph, ball
hypergraph, and $2$-ball hypergraph of a dissimilarity space are either all
hypertrees or none of them are, and whenever they are hypertrees they share
the same underlying trees. Hence, statements $3$, $4$, and $5$ are
equivalent. The result follows.
\end{proof}


\section{Recognition}\label{sec:recognition}

This section presents an algorithm to decide whenever a dissimilarity space is \Trob.
Our approach leverages a key property established in Lemma~\ref{thm:mst} stating that any compatible tree must be a minimum spanning tree of $(X,D)$.
This observation narrows the set of compatible tree candidates and allows us to evaluate each MST individually (Algorithm~\ref{alg:algtrob}).
Building on this property, in Theorem~\ref{thm:rob_recognition} we derive a recognition algorithm for \Trob spaces with time complexity $O(K n^2)$ and space complexity $O(n^2)$, where $n$ denotes the number of elements in the space and $K$ denotes the number of distinct minimum spanning trees of the dissimilarity space.

We begin with the following result, which provides a key structural property of compatible trees. 

\begin{lemma}
\label{thm:mst}
    Let $(X,D)$ be a \emph{\Trob} dissimilarity space. 
    If $T$ is a compatible tree for $(X,D)$, then $T$ is a minimum spanning tree of $(X,D)$.
\end{lemma}
\begin{proof}
    Let $T$ be a compatible tree for a \emph{\Trob} dissimilarity space $(X,D)$.
    Suppose, for the sake of contradiction, that $T$ is not a minimum spanning tree.
    By the cycle property of MSTs, there exists a pair of elements $x,y\in X$ such that $xy$ is not an edge of $T$ and the unique $x-y$-path joining $x$ and $y$ in $T$ contains an edge $vw$ with $D(v,w) > D(x,y)$.
    This contradicts Lemma~\ref{lem:path_characterization}, which requires every edge along a tree path between $x$ and $y$ to have dissimilarity at most $D(x,y)$ whenever $T$ is compatible. 
\end{proof}

Lemma~\ref{thm:mst} provides a necessary condition for compatible trees and restricts the search space to minimum spanning trees. However, the condition is not sufficient: not every MST of a dissimilarity space is compatible, as illustrated in Figure~\ref{fig:tree_evaluation}.

Algorithm~\ref{alg:algtrob} describes a validation procedure that determines whether a given spanning tree $T$ is compatible with a dissimilarity space $(X,D)$. 
The algorithm adopts a dynamic programming strategy in which the Robinson condition is checked on all paths of $T$ in non-decreasing order according to their lengths, where trivial and elementary paths are considered Robinson by definition (Steps 1--7). 
To this end, all non-elementary paths of $T$ are enumerated and grouped by length, retaining only the first two and last two nodes of each path (Step 8); paths are then processed from shortest to longest (Steps 9--15).

The order in which paths are processed is critical for both the correctness and the time complexity of the algorithm: when a path of length $k$ is examined, all its proper sub-paths have already been validated. 
This invariant is exploited by the subroutine PathValidation, described in Algorithm~\ref{alg:subroutine}, which determines whether a path $P = [x_1, \ldots,  x_k]$ satisfies the Robinson condition by inspecting only its endpoints and their immediate neighbors (Steps 4--6). 
Specifically, it tests if all proper sub-paths of $P$ are Robinson together with the inequality:
\[
    D(x_1, x_k) \geq \max\{D(x_1, x_{k-1}),\, D(x_2, x_k)\}.
\]
This local check, which runs in $O(1)$ time, replaces the exhaustive verification over all intermediate triplets in $P$, whose correctness is guaranteed by the inductive hypothesis that all its proper sub-paths have already been tested. 
Furthermore, the algorithm simultaneously counts the number of Robinson paths in $T$, providing a quantitative measure of compatibility that will be used extensively in the subsequent sections.

\begin{lemma}\label{lem:algorithm-validation}
Let $(X,D)$ be a dissimilarity space with $|X|=n$, and let $T$ be a tree on vertex set $X$.
Algorithm~\ref{alg:algtrob} returns $n(n-1)/2$ if and only if $T$ is a compatible tree for $(X,D)$.
\end{lemma}
\begin{proof}
If $T$ is compatible with $(X, D)$, every path satisfies the Robinson condition, and Algorithm~\ref{alg:algtrob}  returns $n(n-1)/2$.
Conversely, if $T$ is not compatible, there exists a path that does not satisfy the Robinson condition.
Let $P = [x_1, \ldots, x_k]$ be such a path of minimum length.
By minimality, every proper sub-path of $P$ is Robinson, so  Lemma~\ref{lem:path_characterization} gives
\[
    D(x_1, x_{k-1}) \geq \max_{1\le i \le k-1} D(x_1, x_i), \qquad 
    D(x_2, x_k) \geq \max_{2 \le j \le k} D(x_j, x_k).
\]
Since $P$ is not Robinson, there exists an internal vertex $y \in P$, with $y\notin\{x_1,x_k\}$, and such that $D(x_1, x_k) < \max\{D(x_1, y), D(y, x_k)\}$; therefore, 
\[
    D(x_1, x_k) < \max\{D(x_1, y),\, D(y, x_k)\} \leq \max\{D(x_1, x_{k-1}),\, D(x_2, x_k)\}.
\]
Hence, the condition tested in PathValidation fails for $P$, and the subroutine correctly returns \texttt{False}, identifying the path as a violation of the Robinson condition. 
Therefore, Algorithm~\ref{alg:algtrob} returns a value strictly smaller than $n(n-1)/2$.
\end{proof}

\begin{algorithm}[tp]
\caption{Subroutine \text{PathValidation}}
\label{alg:subroutine}
\begin{algorithmic}[1]
    \REQUIRE $D$ dissimilarity matrix, 
            $Y$ boolean matrix of valid Robinson paths, 
            and $P$ a path in $T$
    \STATE $output \leftarrow \texttt{False}$
    \STATE $x_1, x_2 \leftarrow$ first and second element in path $P$
    \STATE $x_3, x_4 \leftarrow$ second-to-last and last elements of $P$  
    \IF {$Y[x_1,x_{3}]$ \AND $Y[x_2,x_{4}]$ \AND $D(x_1, x_4) \ge \max\{D(x_1, x_3), D(x_2, x_4)\}$}
        \STATE $output \leftarrow \texttt{True}$ 
    \ENDIF
    \RETURN $output$
    \end{algorithmic}
\end{algorithm}

\begin{algorithm}[tp]
\caption{\Trob Validation}
\label{alg:algtrob}
\begin{algorithmic}[1]
    \REQUIRE $(X,D)$ dissimilarity space on $n$ elements, 
    $T$ spanning tree
    \STATE $paths \leftarrow n-1$ \COMMENT{counter of non-trivial valid Robinson paths}
    \STATE $Y \leftarrow \texttt{False}^{n \times n}$ \COMMENT{boolean matrix of valid Robinson paths}    
    \FOR{$ x \in X$}
        \FOR{$ y \in N[x]$}
            \STATE $Y[x,y]=\texttt{True}$ 
            \COMMENT{trivial and elementary paths}
        \ENDFOR
    \ENDFOR
        \\[3pt]
        \STATE Compute, for each $k=2, \dots, \text{diam}(T)$, the set $\mathcal{P}[k]$ of all paths of length $k$ in $T$, storing for each path the first two and last two vertices
    \\[3pt]
    \FOR{$k = 2$ \TO $\text{diam}(T)$}
        \FOR{$P \in \mathcal{P}[k]$}
            \IF{\text{PathValidation}$(D, Y, P)$}
                \STATE $x,y \leftarrow$ extreme nodes of $P$
                \STATE $Y[x,y] \leftarrow \texttt{True};\, Y[y,x] \leftarrow \texttt{True}$
                \STATE $paths \leftarrow paths+1$ 
            \ENDIF
        \ENDFOR
    \ENDFOR
    \RETURN $paths$
    \end{algorithmic}
\end{algorithm}

\begin{lemma}\label{lem:algorithm2-complexity}
Algorithm~\ref{alg:algtrob} runs in $O(n^2)$ time and requires $O(n^2)$ space, where $n$ denotes the number of elements of the input dissimilarity space.
\end{lemma}
\begin{proof}
The family $\{\mathcal{P}[k]\}_{k\in [2,\operatorname{diam}(T)]}$, storing the first two and last two elements of every paths of length $k$ in $T$, can be constructed in $O(n^2)$ time by computing pairwise distances during a DFS traversal rooted at each vertex of $T$.
Since each of the $O(n^2)$ paths is processed in $O(1)$ time by Algorithm~\ref{alg:subroutine}, the validation of a single tree takes $O(n^2)$ time. 
The overall space is $O(n^2)$, dominated by the matrix $Y$ and the stored path endpoints.
\end{proof}

By combining the validation Algorithm~\ref{alg:algtrob} with an exhaustive examination of the MST space, we derive \Cref{thm:rob_recognition}, stating that \Trob spaces can be recognized with a time complexity $O(Kn^2)$ and space $O(n^2)$, where \(K\) denotes the number of minimum spanning
trees of the dissimilarity space.

\begin{proof}[Proof of \Cref{thm:rob_recognition}]
By Lemma~\ref{thm:mst}, any tree compatible with $(X,D)$ must be a minimum spanning tree.
The recognition algorithm therefore enumerates all MSTs of $(X,D)$ and validates each one
using Algorithm~\ref{alg:algtrob}, accepting if and only if at least one compatible tree is found.

The enumeration of all MSTs proceeds in two stages.
In the first stage, an auxiliary graph is constructed  whose spanning trees are in bijection
with the MSTs of $(X,D)$~\cite{eppstein1995representing}. 
The auxiliary graph can be built in $O(m+n\log n)$ time and $O(n^2)$ space, where $m$ is the number of edges in the former graph ($m=n(n-1)/2$ in our case). 
In the second stage, the procedure described in~\cite{Kapoor1995} enumerates all $K$ spanning trees of the
auxiliary graph, thereby recovering all $K$ MSTs of $(X,D)$. 
This enumeration can be performed in $O(m + n\log n + K) = O(n^2 + K)$ time to describe all spanning trees implicitly, with an additional $O(Kn)$ time to output each MST explicitly, and uses $O(n^2)$ space throughout.

Finally, each of the $K$ MSTs is validated by Algorithm~\ref{alg:algtrob} in $O(n^2)$ time, yielding an  overall time and space complexity of $O(Kn^2)$ and $O(n^2)$, respectively.
\end{proof}

Since the number of MTSs can be exponential in the number of elements in the worst case, our method does not universally outperform the existing alternative with complexity $O(n^5)$~\cite{brucker2005hypertrees}. 
Nevertheless, our approach represents a significant improvement when the number of MSTs is small ($K = o(n^3)$). 
This condition is satisfied in numerous practically relevant settings and many real-world network datasets where dissimilarity values are predominantly distinct.

\section{Optimization approach}
\label{sec:optimization_approach}

Now, we introduce an optimization framework designed to bridge the gap between arbitrary dissimilarity spaces and exact \Trob spaces. Because empirical data rarely exhibit perfect structural compatibility, not all dissimilarity matrices inherently admit a \Trob spanning tree. Therefore, we aim to find a spanning tree that maximizes structural adherence to the \Trob property. Next, Section \ref{subsec:problem_def} formalizes this objective by defining the combinatorial optimization problem and illustrating it with a minimal working example. Subsequently, Section \ref{subsec:algorithmic_proposal} details our algorithmic proposal to efficiently navigate the search space of spanning trees.

\subsection{Problem formulation}
\label{subsec:problem_def}

As established in Definition \ref{def:trob}, a dissimilarity space $(X,D)$ is \Trob if it admits a compatible tree where every path induces a Robinson subspace. However, empirical data rarely exhibits perfect structural compatibility. To quantify the structural adherence of any candidate spanning tree $T \in \mathcal{T}(X)$ to the \Trob property, we introduce the objective function $PR_D(T)$. 

Motivated by Lemma \ref{lem:path_characterization}, $PR_D(T)$ evaluates the structural affinity by counting the number of valid Robinsonian paths in $T$. The optimization problem consists of finding a spanning tree $T^*$ that maximizes this path count:
\begin{equation}
    T^* = \underset{T \in \mathcal{T}(X)}{\arg\max} \; PR_D(T).
\end{equation}

By Lemma \ref{lem:algorithm-validation}, the theoretical maximum for $PR_D(T)$ is $n(n-1)/2$, which is strictly attained if and only if $(X,D)$ is exactly \Trob and $T$ is its compatible tree.

To illustrate this objective and highlight why finding $T^*$ is non-trivial, even when restricting the search to the space of exact minimum spanning trees (as mandated by \Cref{thm:mst}), consider the following dissimilarity matrix:
\begin{equation}
D = \begin{bmatrix} 
0 & 1 & 1 & 1 \\ 
1 & 0 & 1 & 2 \\ 
1 & 1 & 0 & 1 \\ 
1 & 2 & 1 & 0 
\end{bmatrix}.
\end{equation}

For this specific dissimilarity space $(X,D)$, Figure \ref{fig:tree_evaluation} presents two candidate minimum spanning trees, $T_1$ and $T_2$, both having an identical minimal weight of $w(T)=3$.

\begin{figure}[t]
    \centering
    \begin{tikzpicture}[
        scale=1.1,
        node/.style={circle, draw=black, thick, fill=white, minimum size=6.5mm, inner sep=0pt},
        tree edge/.style={ultra thick, black},
        valid path/.style={thick, dashed, draw=blue!80!black},
        invalid path/.style={thick, dashed, draw=red!80!black},
        lbl/.style={midway, fill=white, inner sep=2pt, font=\scriptsize, text=black}
    ]
        
        \begin{scope}[shift={(0,0)}]
            \node[node] (x1) at (0, 1.5) {$x_1$};
            \node[node] (x2) at (-1.8, 0) {$x_2$};
            \node[node] (x4) at (0, 0) {$x_4$};
            \node[node] (x3) at (1.8, 0) {$x_3$};
            
            \draw[tree edge] (x1) -- (x2);
            \draw[tree edge] (x1) -- (x4);
            \draw[tree edge] (x1) -- (x3);
            
            \draw[valid path, bend left=25] (x2) to node[lbl] {$1$} (x1);
            \draw[valid path, bend right=25] (x4) to node[lbl] {$1$} (x1);
            \draw[valid path, bend right=25] (x3) to node[lbl] {$1$} (x1);
            \draw[valid path, bend right=40] (x2) to node[lbl] {$2$} (x4);
            \draw[valid path, bend right=40] (x4) to node[lbl] {$1$} (x3);
            \draw[valid path, bend right=65] (x2) to node[lbl] {$1$} (x3);
            
            \node[draw=none] at (0, -1.6) {$T_1$};
        \end{scope}

        \begin{scope}[shift={(6.8,0)}]
            \node[node] (y2) at (-2.4, 0) {$x_2$};
            \node[node] (y1) at (-0.8, 0) {$x_1$};
            \node[node] (y4) at (0.8, 0) {$x_4$};
            \node[node] (y3) at (2.4, 0) {$x_3$};
            
            \draw[tree edge] (y2) -- (y1);
            \draw[tree edge] (y1) -- (y4);
            \draw[tree edge] (y4) -- (y3);
            
            \draw[valid path, bend left=40] (y2) to node[lbl] {$1$} (y1);
            \draw[valid path, bend left=40] (y1) to node[lbl] {$1$} (y4);
            \draw[valid path, bend left=40] (y4) to node[lbl] {$1$} (y3);
            \draw[valid path, bend right=45] (y2) to node[lbl] {$2$} (y4);
            \draw[valid path, bend right=45] (y1) to node[lbl] {$1$} (y3);
            \draw[invalid path, bend left=60] (y2) to node[lbl, text=red!80!black, font=\scriptsize] {$1$} node[above=5pt, text=red!80!black, font=\scriptsize] {$D(x_2,x_4) > D(x_2,x_3)$} (y3);
            
            \node[draw=none] at (0, -1.6) {$T_2$};
        \end{scope}

        \begin{scope}[shift={(1.9,-3)}]
            \draw[tree edge] (-3.7, 0.5) --++ (0.8, 0) node[right, font=\small, text=black, inner sep=4pt] {Tree edge};
            \draw[valid path] (-.7, 0.5) --++ (0.8, 0) node[right, font=\small, text=black, inner sep=4pt] {Valid Robinsonian path};
            \draw[invalid path] (4, 0.5) --++ (0.8, 0) node[right, font=\small, text=black, inner sep=4pt] {\Trob violation};
        \end{scope}
    \end{tikzpicture}
\caption{Evaluation of \Trob compatibility for spanning trees $T_1$ and $T_2$, where dashed arcs represent pairwise dissimilarities and highlight the structural violation in $T_2$.}
    \label{fig:tree_evaluation}
\end{figure}
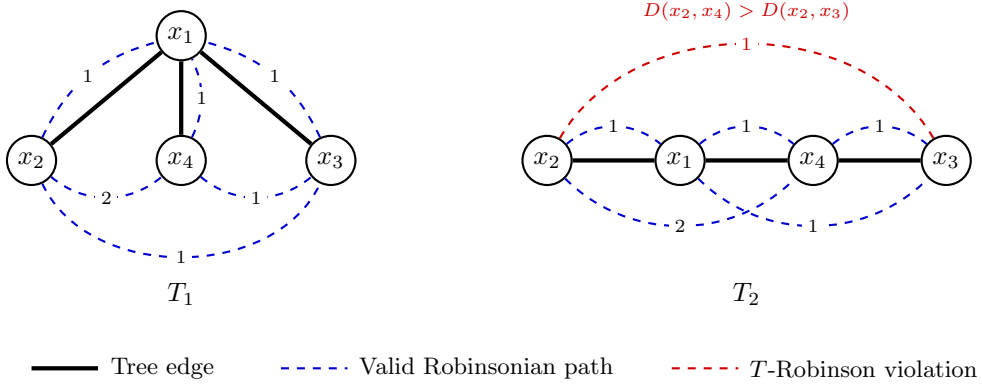

Evaluating the objective function over both candidates highlights their structural differences. The star graph $T_1$ satisfies the Robinson constraint across all six of its unique paths. For instance, the path $P_{x_2 x_4} = [x_2, x_1, x_4]$ connects endpoints with a dissimilarity of $D(x_2, x_4) = 2$; its internal connections yield distances $D(x_1, x_2) = 1$ and $D(x_1, x_4) = 1$, strictly adhering to the condition $D(u,v) \le D(x_2, x_4)$. Thus, $PR_D(T_1) = 6$, achieving the global maximum.

In contrast, the path graph $T_2$ yields $PR_D(T_2) = 5$. While five of its subpaths are valid, the maximal path $P_{x_2 x_3}$ violates the Robinson condition because evaluating the intermediate node $x_4$ reveals a distance of $D(x_2, x_4) = 2$, which strictly exceeds the endpoint dissimilarity $D(x_2, x_3) = 1$. 

This confirms the fact that minimizing edge weight is a necessary but insufficient condition for maximizing the \Trob property. Identifying an optimal tree requires navigating the highly degenerate space of MSTs, which motivates the algorithmic proposal detailed in the following subsection.

\subsection{Algorithmic proposal}
\label{subsec:algorithmic_proposal}

Given the combinatorial nature of the problem defined in Section \ref{subsec:problem_def}, finding the exact \Trob spanning tree requires navigating the complete space of $\mathcal{T}(X)$. Since the number of spanning trees for a complete graph on $n$ vertices is $n^{n-2}$, 
exact algorithms become computationally prohibitive for large-scale instances.

To tackle this problem, we propose a multistart heuristic framework. Multistart methods are highly effective for rugged objective landscapes, as they iteratively balance exploration (diversification) and exploitation (intensification) and have been proposed successfully for similar problems \cite{cavero2025graph}. Our framework operates by first generating a diverse, high-quality pool of initial trees. Then, it sequentially applies a local search phase to refine the best available topologies, maximizing their structural compatibility until a computational budget is exhausted. 

Algorithm \ref{alg:multistart_framework} outlines this proposed approach, taking the dissimilarity matrix $D$ and a global time limit $t_{max}$ as inputs. After initializing the incumbent solution $T^*$ and its objective value $PR_{best}$ (Steps \ref{step:init_tree}--\ref{step:init_pr}), the framework executes the construction phase to build and sort a pool of candidate trees, $\mathcal{T}_{pool}$ (Step \ref{step:constructpool}). The search is then driven by a time-bounded loop (Step \ref{step:while_loop}). In each iteration $i$, the algorithm selects the $i$-th best tree from the pool as the initial solution $T_{init}$ (Step \ref{step:select}) and applies a local search to refine it into a local optimum $T_{local}$ (Step \ref{step:localsearch}). If $T_{local}$ strictly improves upon the incumbent objective, $T^*$ is updated accordingly (Steps \ref{step:check_best}--\ref{step:update_pr}). This sequential improvement process repeats until the $t_{max}$ budget is consumed or the pool is exhausted, returning the highest-quality \Trob approximation found.

\begin{algorithm}[t]
\caption{Multistart Algorithm}
\label{alg:multistart_framework}
\begin{algorithmic}[1]
\REQUIRE Dissimilarity matrix $D$, maximum execution time $t_{max}$
\ENSURE Spanning tree $T^*$ approximating the optimal \Trob compatibility
\STATE $T^* \leftarrow \emptyset$ \label{step:init_tree}
\STATE $PR_{best} \leftarrow -1$ \label{step:init_pr}
\STATE $\mathcal{T}_{pool} \leftarrow \text{ConstructionPhase}(D)$ \label{step:constructpool}
\STATE $i \leftarrow 1$
\WHILE{elapsed time $< t_{max}$ \textbf{and} $i \le |\mathcal{T}_{pool}|$} \label{step:while_loop}
    \STATE $T_{init} \leftarrow \mathcal{T}_{pool}[i]$ \label{step:select}
    \STATE $T_{local} \leftarrow \text{LocalSearch}(T_{init}, D)$ \label{step:localsearch}
    \IF{$PR_D(T_{local}) > PR_{best}$} \label{step:check_best}
        \STATE $T^* \leftarrow T_{local}$ \label{step:update_tree}
        \STATE $PR_{best} \leftarrow PR_D(T_{local})$ \label{step:update_pr}
    \ENDIF
    \STATE $i \leftarrow i + 1$
\ENDWHILE
\RETURN $T^*$ \label{step:return}
\end{algorithmic}
\end{algorithm}

\subsubsection*{Construction phase}
\label{subsubsec:construction}

The construction phase (Algorithm \ref{alg:multistart_framework}, Step \ref{step:constructpool}) is responsible for generating the diverse pool of candidate spanning trees, $\mathcal{T}_{pool}$, which provides the starting points for the subsequent local search. This procedure exploits the structural correlation between MSTs and \Trob spaces through a combination of exact evaluations and stochastic perturbations.

First, relying on the necessary condition established in \Cref{thm:mst}, which proves that any compatible tree must be a minimum spanning tree of $(X,D)$, we explore the set of exact MSTs for the dissimilarity space. Because empirical data often contains tied dissimilarity values, the optimization landscape frequently exhibits a high multiplicity of optimal solutions. We generate a subset of exact minimum weight trees within a predefined computational time limit $t_{mst}$, appending all unique topologies to our candidate pool. 

To escape potentially suboptimal topologies inherent to exact MSTs, the second stage employs a noise-induced exploration strategy. We apply random perturbations to the original matrix, thereby generating novel MST topologies. Specifically, we compute a perturbed matrix $D_{noise} = D + \varepsilon R$, where $R \in \mathbb{R}^{n \times n}$ is a symmetric random matrix with a zero diagonal whose strictly upper-triangular entries are sampled from a standard normal distribution $\mathcal{N}(0,1)$, and $\varepsilon > 0$ dictates the noise intensity. This perturbation is applied independently $n_{pert}$ times, and one MST of each modified space $(X, D_{noise})$ is computed and added to the candidate pool. 

Crucially, while the perturbed matrices dictate the generated tree topologies, their structural compatibility is exclusively evaluated using the original dissimilarity matrix $D$. Once both exact and stochastically perturbed trees are gathered, we evaluate the objective function $PR_D(T)$ for every tree in the pool. Finally, $\mathcal{T}_{pool}$ is sorted in descending order according to these objective values. Consequently, during the multistart execution, the algorithm systematically prioritizes the refinement of the most promising topological structures.

\subsubsection*{Local search phase}
\label{subsubsec:local_search}

The initial trees generated during the construction phase act as starting points for a systematic neighborhood exploration. The local search phase (Algorithm \ref{alg:multistart_framework}, Step \ref{step:localsearch}) is designed to iteratively refine a candidate tree $T$ by applying structural perturbations that increase its objective value $PR_D(T)$ until a local maximum is reached.

To formally define the neighborhood of a solution, let $T = (X, E)$ be a candidate spanning tree and let $L(T) \subset X$ denote the set of its leaf vertices (i.e., vertices with a degree of exactly $1$). We define a fundamental structural operation denoted \textit{leaf relocation}. For any given leaf $v \in L(T)$ adjacent to a unique vertex $u$, the operation consists of deleting the edge $vu$ and inserting a new edge $vw$, where $w \in X \setminus \{v, u\}$. Because this operation merely reattaches a terminal node to a different part of the existing structure, it strictly preserves the connectivity and acyclic properties required for a spanning tree.

Based on this operation, $T$'s neighborhood $\mathcal{N}(T)$ is defined as the set of all valid spanning trees that can be reached from $T$ by applying exactly one leaf relocation. Mathematically, this is expressed as:
\begin{equation*}
    \mathcal{N}(T) = \Big\{ \big(X, (E \setminus \{ uv\} )\cup \{vw\}\big) \;\Big|\; v \in L(T), \; uv \in E, \; w \in X \setminus \{u,v\} \Big\}.
\end{equation*}

To illustrate this neighborhood generation, consider the candidate tree $T_2$ introduced in Section \ref{subsec:problem_def}, which induces the path graph $x_2 - x_1 - x_4 - x_3$. The set of leaves is $L(T_2) = \{x_2, x_3\}$. Initially, $x_2$ is adjacent to $x_1$. By removing the edge $x_2 x_1$, we can reattach $x_2$ to any other vertex in the set $\{x_4, x_3\}$. As shown in Figure \ref{fig:neighborhood_example}, relocating $x_2$ to $x_4$ generates the neighbor $T_{2a}$, while relocating it to $x_3$ generates the neighbor $T_{2b}$. Applying this same logic to the other leaf $x_3$ generates two additional neighbors, resulting in a total neighborhood size of $|\mathcal{N}(T_2)| = 4$. 

\begin{figure}[t]
    \centering
    \begin{tikzpicture}[
        scale=1.1,
        node/.style={circle, draw=black, thick, fill=white, minimum size=6.5mm, inner sep=0pt},
        tree edge/.style={ultra thick, black},
        new edge/.style={ultra thick, solid, draw=green!60!black},
        deleted edge/.style={ultra thick, dotted, draw=red!40!white}
    ]
        
        \begin{scope}[shift={(3.5, 3)}]
            \node[node] (y2) at (-2.4, -0.2) {$x_2$};
            \node[node] (y1) at (-0.8, -0.2) {$x_1$};
            \node[node] (y4) at (0.8, -0.2) {$x_4$};
            \node[node] (y3) at (2.4, -0.2) {$x_3$};
            
            \draw[tree edge] (y2) -- (y1);
            \draw[tree edge] (y1) -- (y4);
            \draw[tree edge] (y4) -- (y3);
            
            \node[draw=none] at (0, 0.8) {Current tree $T_2$};
            \node[draw=none, font=\small, text=gray] at (0, 0.4) {Leaf to relocate: $x_2$};
            
            \draw[->, thick, shorten >= 5pt, shorten <= 5pt] (-0.8, -0.75) -- 
                node[midway, above left, font=\small, inner sep=2pt] {Relocate $x_2$ to $x_4$} (-2, -1.5);
            \draw[->, thick, shorten >= 5pt, shorten <= 5pt] (0.8, -0.75) -- 
                node[midway, above right, font=\small, inner sep=2pt] {Relocate $x_2$ to $x_3$} (2, -1.5);
        \end{scope}

        \begin{scope}[shift={(0,0)}]
            \node[node] (n1_1) at (-1.5, 0) {$x_1$};
            \node[node] (n1_4) at (0, 0) {$x_4$};
            \node[node] (n1_3) at (1.5, 0) {$x_3$};
            \node[node] (n1_2) at (0, 1.3) {$x_2$};
            
            \draw[deleted edge] (n1_2) -- (n1_1); 
            \draw[new edge] (n1_2) -- (n1_4);     
            \draw[tree edge] (n1_1) -- (n1_4);
            \draw[tree edge] (n1_4) -- (n1_3);
            
            \node[draw=none] at (0, -0.8) {Neighbor $T_{2a}$};
        \end{scope}

        \begin{scope}[shift={(7,0)}]
            \node[node] (n2_1) at (-1.5, 0) {$x_1$};
            \node[node] (n2_4) at (0, 0) {$x_4$};
            \node[node] (n2_3) at (1.5, 0) {$x_3$};
            \node[node] (n2_2) at (1.5, 1.3) {$x_2$};
            
            \draw[deleted edge, bend right=15] (n2_2) to (n2_1); 
            \draw[new edge] (n2_2) -- (n2_3);     
            \draw[tree edge] (n2_1) -- (n2_4);
            \draw[tree edge] (n2_4) -- (n2_3);
            
            \node[draw=none] at (0, -0.8) {Neighbor $T_{2b}$};
        \end{scope}

    \end{tikzpicture}
\caption{Neighborhood generation $\mathcal{N}(T_2)$ for leaf $x_2$. Removing the edge $x_2x_1$ (dotted red) and reattaching $x_2$ to vertices $x_4$ and $x_3$ yields the neighbors $T_{2a}$ and $T_{2b}$ (solid green), respectively.}
    \label{fig:neighborhood_example}
\end{figure}
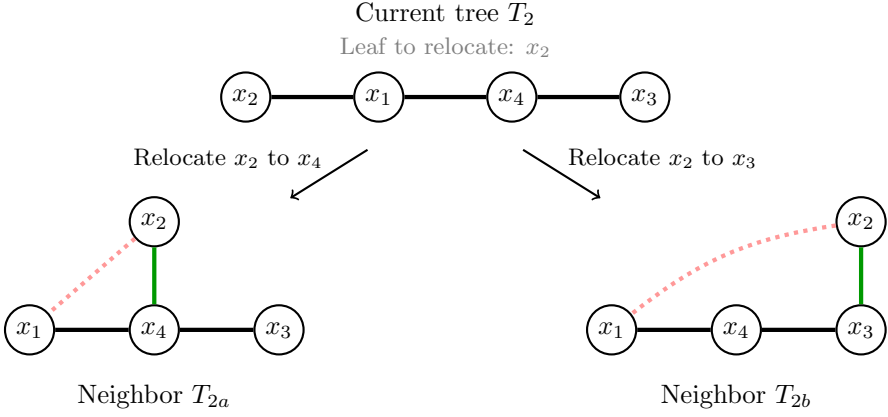

To traverse this neighborhood, two distinct exploration strategies will be analyzed and compared in our computational framework: \textit{first improvement}, which immediately accepts the first neighbor that strictly increases $PR_D(T)$, and \textit{best improvement}, which evaluates the entire neighborhood $\mathcal{N}(T)$ to select the neighbor yielding the maximum objective gain.

\section{Experimental analysis}
\label{sec:experiments}

To evaluate the performance, scalability, and practical applicability of the proposed multistart framework, we conduct a comprehensive computational analysis. The experiments are divided into two main categories: first, an algorithmic benchmarking on synthetic instances to assess the individual contribution of the construction and local search strategies; second, an application to real-world biological datasets to demonstrate the framework's utility in empirical scenarios.

All algorithms were implemented in Python 3.9  and the computational experiments were executed on a dedicated machine equipped with a Intel Xeon Gold 6226R 2.90 GHz with 8 GB RAM.

\subsection{Analysis of the algorithm}
\label{subsec:analysis_algorithm}

The algorithmic components were evaluated over a controlled set of synthetic dissimilarity spaces. The instances consist of complete, symmetric dissimilarity matrices of size $n \in \{7, 10, 15, 20, 25, 30\}$. To avoid algorithmic bias, these matrices were populated using a uniform random linear generator \cite{harris2020array}, ensuring that they lack any inherent or pre-planted Robinsonian structure. This represents the most challenging scenario (the worst-case topology) for the heuristic, as the optimization landscape is highly irregular.

To facilitate comparison across varying matrix sizes, the objective function reported in this section is the normalized \Trob compatibility, defined as the ratio of valid Robinsonian paths to the maximum possible paths in a spanning tree: $\overline{PR}_D(T) = \frac{2 \cdot PR_D(T)}{n(n-1)}$. Thus, $\overline{PR}_D(T) \in [0,1]$.

The first experiment isolates the construction phase (Section \ref{subsubsec:construction}) to quantify the value of the noise-induced perturbation strategy compared to a purely deterministic exploration of exact MSTs. Table \ref{tab:construction_results} reports the computational results for the construction phase. A clear trend is that the structural quality of exact MSTs deteriorates as the instance size $n$ increases. Although the MST space is highly degenerate, averaging $79.4$ minimal trees for $n=30$, relying solely on tie-breaking rules does not generate sufficient topological diversity to capture the Robinsonian property in larger spaces. The addition of random noise addresses this issue. By perturbing the dissimilarity matrix, the ordinal ranking of the edges changes, forcing the algorithm to sample spanning trees with fundamentally different topologies. The perturbation budget was fixed at $200$ iterations, as preliminary computational tests revealed that the structural improvement exhibits asymptotic convergence beyond this threshold. As shown in the table, this strategy yields an average improvement of $5.49\%$ over the deterministic approach, providing a stronger set of initial solutions for the local search phase.

\begin{table}[htbp]
\centering
\caption{Performance comparison within the construction phase. The table contrasts the average normalized objective $\overline{PR}_D(T)$ obtained solely from exact MSTs against the enhanced pool generated using stochastic perturbations (MST + Noise).}
\label{tab:construction_results}
\begin{adjustbox}{max width=\textwidth}
\begin{tabular}{@{}lccccc@{}}
\toprule
 & \multicolumn{2}{c}{\textbf{Exact MST}} & \multicolumn{3}{c}{\textbf{MST + Noise}} \\ \cmidrule(lr){2-3} \cmidrule(l){4-6}
\textbf{Size ($n$)} & $\overline{PR}_D(T)$ & \textbf{\# Trees Explored} & $\overline{PR}_D(T)$ & \textbf{\# Trees Explored} & \textbf{Improvement (\%)} \\ \midrule
7  & 0.786 & \phantom{1}1.8  & 0.810 & 200 & 3.12 \\
10 & 0.672 & \phantom{1}5.4  & 0.719 & 200 & 6.83 \\
15 & 0.580 & 15.0 & 0.610 & 200 & 4.35 \\
20 & 0.466 & 10.2 & 0.493 & 200 & 4.49 \\
25 & 0.378 & 66.9 & 0.418 & 200 & 8.26 \\
30 & 0.349 & 79.4 & 0.379 & 200 & 5.86 \\ \midrule
\textbf{Average} & \textbf{0.538} & \textbf{29.8} & \textbf{0.572} & \textbf{200} & \textbf{5.49} \\ \bottomrule
\end{tabular}
\end{adjustbox}
\end{table}

Building upon the initial pool generated by the construction phase, which evaluates exact MSTs and applies the aforementioned $200$ noise-induced iterations, we assess the impact of the neighborhood exploration strategies: first improvement (FI) and best improvement (BI). The algorithms were executed with a strict global time budget of $t_{max} = 600$ seconds per instance. 

Table \ref{tab:aggregate_performance} summarizes the aggregate performance across all synthetic instances. The construction phase reliably provides a strong initial baseline, securing an average compatibility score of $0.573$ within approximately $11$ seconds. Applying a single iteration of the local search (LS) until a local optimum is reached significantly increases the metric across both strategies. However, FI achieves a slightly better local optimum ($0.748$) in less computational time ($266.11$ seconds) compared to BI ($0.743$ in $316.63$ seconds). This efficiency directly scales to the full multistart (MS) execution: because FI consumes less time per LS, it completes more global iterations within the 600-second budget, ultimately converging to a superior global metric ($0.751$ versus $0.744$).
\begin{table}[htbp]
\centering
\caption{Aggregate performance of the FI and BI strategies across the construction, single local search (LS), and complete multistart (MS) phases. Times are reported in seconds. The global time limit was set to 600 seconds.}
\label{tab:aggregate_performance}
\begin{tabular}{@{}lcccccc@{}}
\toprule
\textbf{Strategy} & \multicolumn{2}{c}{\textbf{Construction}} & \multicolumn{2}{c}{\textbf{Local Search}} & \multicolumn{2}{c}{\textbf{Multistart}} \\ \cmidrule(lr){2-3} \cmidrule(lr){4-5} \cmidrule(l){6-7} 
                  & $\overline{PR}_D(T)$ & Time (s)     & $\overline{PR}_D(T)$ & Time (s)                 & $\overline{PR}_D(T)$ & Time (s)             \\ \midrule
BI                & 0.573                 & 10.91       & 0.743                 & 316.63                  & 0.744                 & 600.01              \\
FI                & 0.573                 & 10.91       & \textbf{0.748}        & \textbf{266.11}         & \textbf{0.751}        & 600.01              \\ \midrule
\textbf{Average}  & 0.573                 & 10.91       & 0.746                 & 291.37                  & 0.748                 & 600.01              \\ \bottomrule
\end{tabular}
\end{table}

Analyzing the scalability of these strategies as a function of the instance size $n$ reveals a steep exponential growth in the computational time required for neighborhood exploration. While the construction time remains highly stable across all dimensions (fluctuating tightly between $10$ and $13$ seconds), the LS phase rapidly becomes the bottleneck. For smaller topologies ($n \le 10$), both FI and BI efficiently resolve the neighborhood within seconds. However, as $n$ increases, the computational overhead of evaluating the entire neighborhood $\mathcal{N}(T)$ heavily penalizes the BI strategy. For instance, at $n=15$, BI requires more than double the time of FI ($206.8$ seconds versus $81.9$ seconds). By $n=30$, the discrete search space is so vast that identifying a single local optimum exhausts the entire 600-second budget, forcing both strategies to systematically hit the timeout restriction before a full MS cycle can be completed.

\subsection{Application to real-world biological datasets}
\label{subsec:exp_real_data}

To assess the practical applicability of the proposed optimization framework, we evaluated its performance on real-world genomic datasets. In bioinformatics, estimating the evolutionary divergence between species or the transcriptomic affinity between samples inherently relies on pairwise sequence dissimilarities. Mapping these relationships into a \Trob compatible spanning tree provides a natural, hierarchical representation of the data without enforcing rigid clustering assumptions. 

\begin{figure}[ht!]
    \centering
    \includegraphics[width=0.55\textwidth]{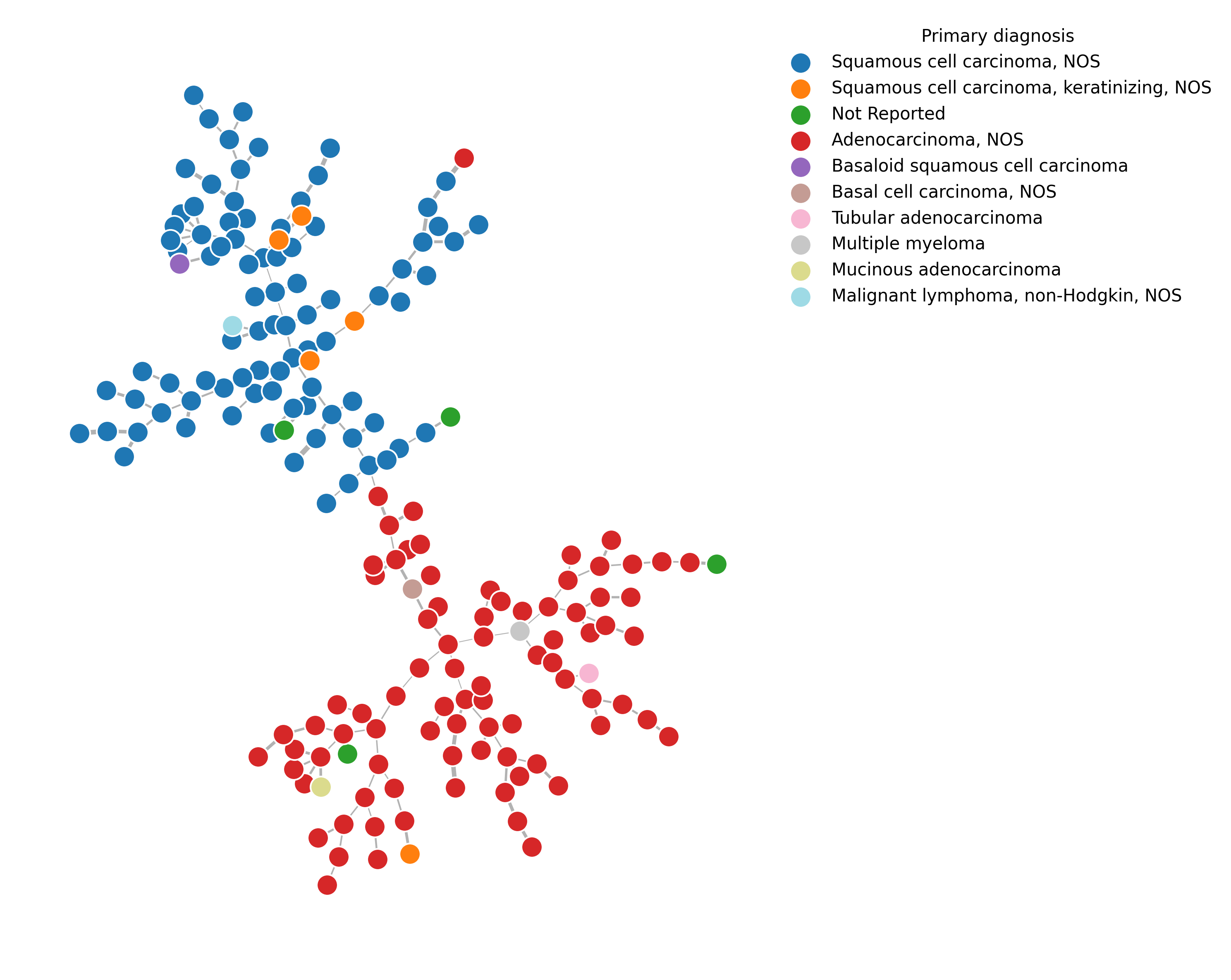}
    \vspace{0.4cm}
    
    \includegraphics[width=0.55\textwidth]{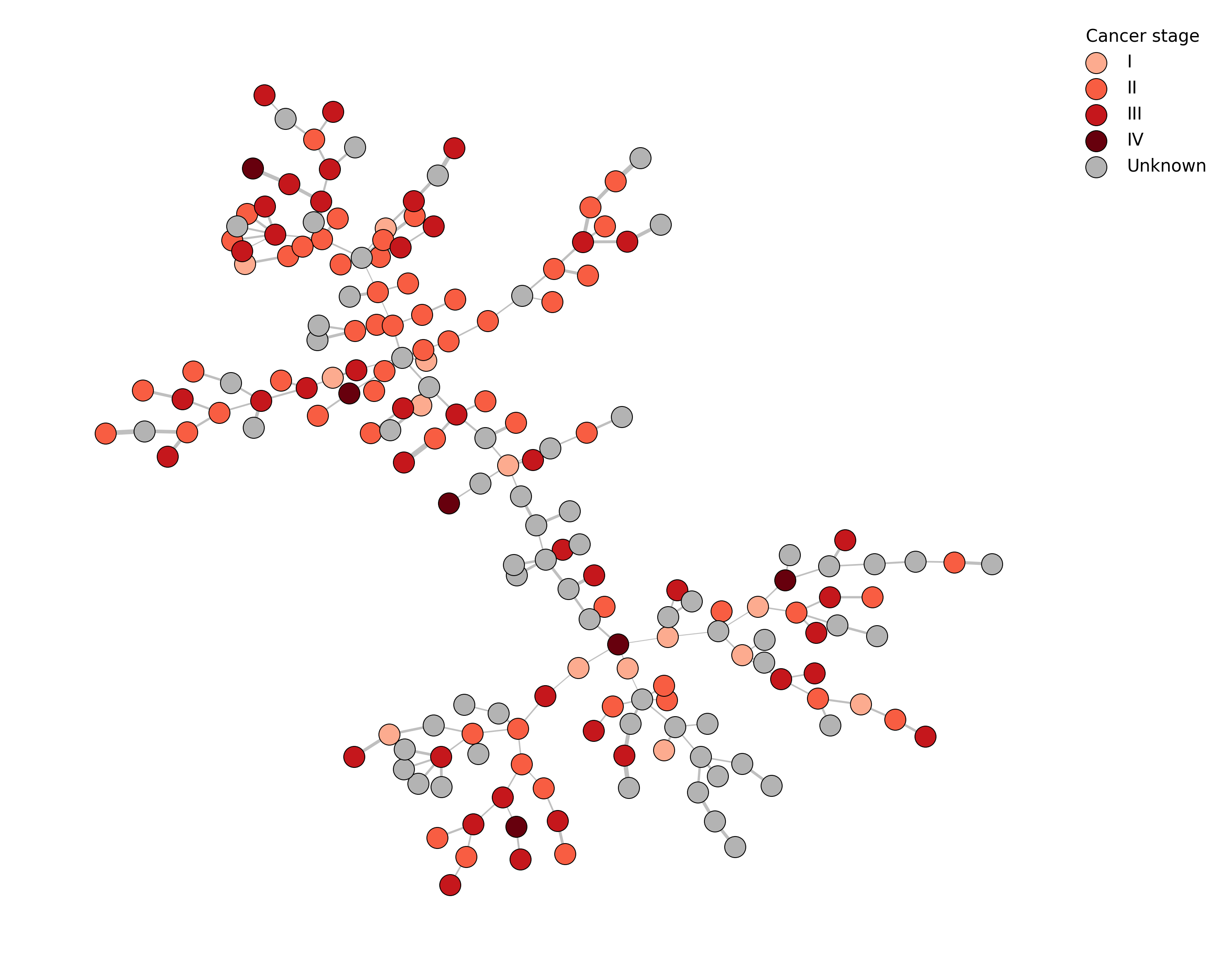}
    \vspace{0.4cm}
    
    \includegraphics[width=0.55\textwidth]{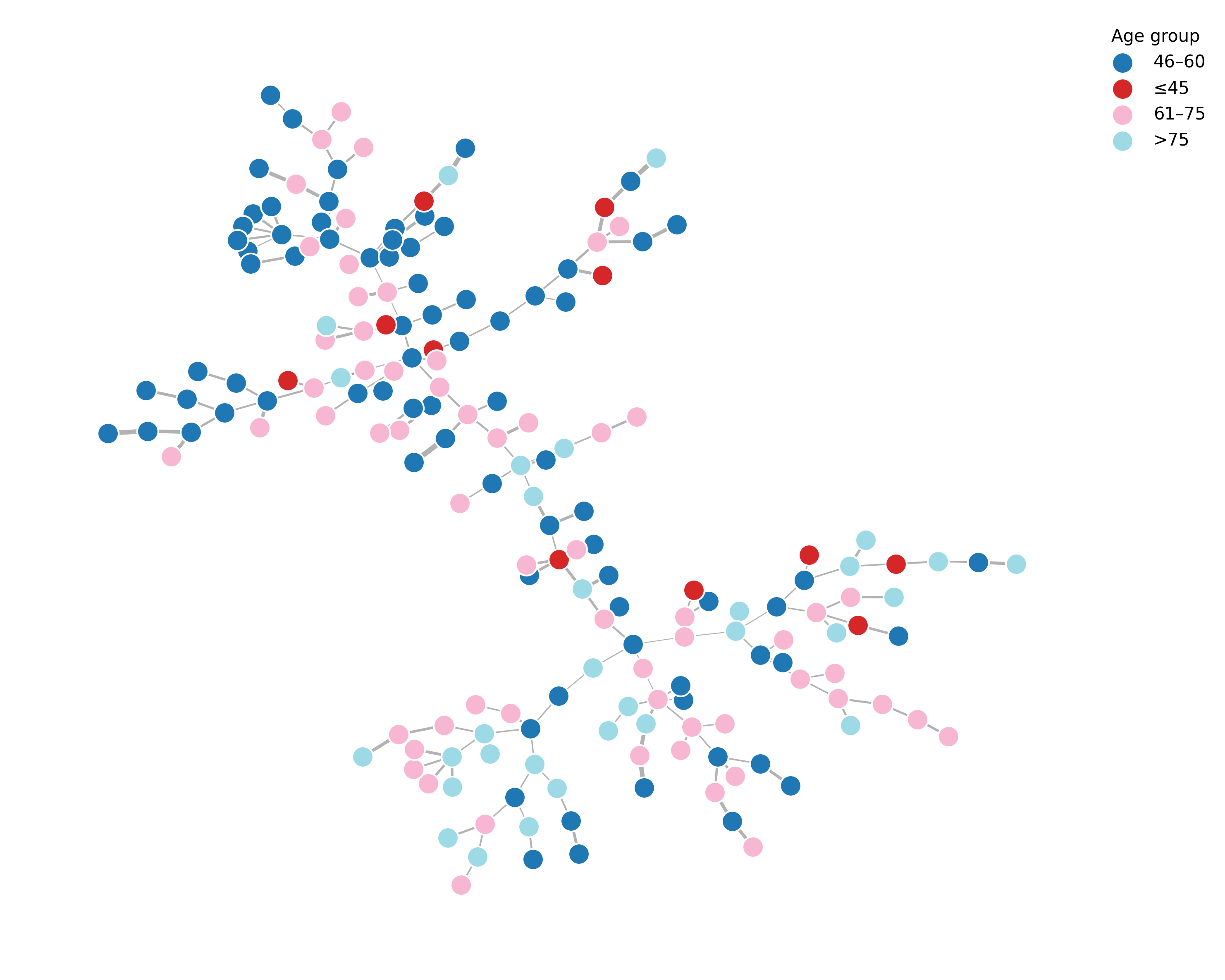}
    
    \caption{Visualization of the optimized \Trob spanning tree for the TCGA-ESCA transcriptomics dataset. The identical tree structure is color-coded by \textbf{Top:} Primary diagnosis, showing a clear transcriptomic bifurcation. \textbf{Middle:} Cancer stage, presenting a more heterogeneous distribution. \textbf{Bottom:} Patient age group, revealing distinct demographic gradients across the main branches.}
    \label{fig:tcga_esca_comparison}
\end{figure}

For our experimental evaluation, we compiled a diverse suite of biological instances categorized into three primary datasets:

\begin{enumerate}
    \item {NCBI genomes:} We utilized genomic sequences from the RefSeq and GenBank repositories \cite{sayers2025database}. To quantify structural relationships without full sequence alignment, sequence profiles were extracted using $k$-mers with a word length of $k=31$. The dissimilarity between any pair of genomes was computed using the Jaccard distance over their respective $k$-mer sets. We extracted a representative subset of $1,000$ genomes across distinct taxonomic domains (viruses, archaea, and bacteria), yielding a $1000 \times 1000$ symmetric dissimilarity matrix bounded in $[0,1]$.
    
    \item {TCGA transcriptomics:} This set comprises primary tumor samples from \textit{The Cancer Genome Atlas} (TCGA) \cite{grossman2016toward}, specifically the Esophageal Carcinoma (TCGA-ESCA, $n=184$) and Rectum Adenocarcinoma (TCGA-READ, $n=166$) cohorts. Gene expression profiles (RNA-Seq) were normalized via the Trimmed Mean of M-values (TMM) method and transformed to log2-CPM. Retaining the $1,000$ most highly variable genes, dissimilarity was defined as the complement of the Spearman correlation coefficient, $D(x,y) = 1 - \rho(x,y)$, yielding values bounded in $[0,2]$.
    
    \item {Fungal protein families:} This dataset captures pairwise dissimilarity computed as the percentage of differing positions over two aligned genome sequences for five fungal protein families (e.g., Laccases, Peroxidases, Chitinases) distributed across Polyporales and Agaricales genomes \cite{arias_soto_2026_21012487}, providing dissimilarity matrices ranging from $n=57$ to $n=126$.
\end{enumerate}

Unlike synthetic instances, real-world biological data is heavily burdened by evolutionary noise, homoplasy, and tied distances, naturally deviating from a perfect hierarchical structure. This structural degradation is immediately evident in the exact minimum spanning trees. Across all evaluated biological instances, the exact MSTs yielded a low average \Trob compatibility of $\overline{PR}_D(T) \approx 0.136$.  Under a restricted 1-hour budget, the heuristic consistently optimized the tree topologies, elevating the average structural metric to $\overline{PR}_D(T) \approx 0.224$. This represents a substantial relative improvement of over $65.10\%$ compared to the deterministic MST baseline.

Next, to visually corroborate that maximizing the \Trob property inherently captures meaningful biological relationships, we mapped sample metadata onto the optimized topologies. Figure \ref{fig:tcga_esca_comparison} illustrates the optimized \Trob spanning tree for the TCGA-ESCA cohort, incorporating primary diagnosis, cancer stage, and patient age group.

In the top panel, nodes are color-coded according to their primary diagnosis. As observed, minimizing the \Trob violations forces the tree topology to naturally group similar transcriptomic profiles into distinct, contiguous branches, effectively organizing the tumors by their underlying genetic profile without requiring predefined clustering algorithms. Most notably, a clear binary division emerges within the structure, accurately segregating squamous cell carcinomas (upper branch) from adenocarcinomas (lower branch). Furthermore, four cases lacking a reported primary diagnosis (green nodes) are seamlessly integrated into these established clusters based purely on their transcriptomic dissimilarities, demonstrating the algorithm's capacity to uncover hidden biological relationships and properly classify unannotated elements.

The middle and bottom panels display the identical tree structure color-coded by clinical variables. While the cancer stage (middle panel) exhibits a scattered, non-contiguous distribution—confirming that the \Trob formulation strictly adheres to latent genetic similarities rather than arbitrary clinical staging—the age group mapping (bottom panel) reveals a noteworthy macroscopic trend. Patients aged 46 to 60 years are predominantly localized within the upper branch, coinciding with the squamous cell carcinoma cluster, whereas older patients ($\geq$ 61 years) are mostly concentrated in the lower adenocarcinoma branch. This underlying spatial distribution accurately reflects the distinct demographic characteristics inherently associated with these two major esophageal cancer subtypes.

\section{Future work}
\label{sec:conclusions}

This work opens several promising avenues for future research. From an algorithmic perspective, an important open question is whether the recognition of \Trob spaces can be performed more efficiently. Our recognition algorithm relies on the enumeration of minimum spanning trees, and reducing or eliminating this dependency remains a challenging problem. In particular, it would be interesting to investigate whether the techniques recently introduced by Préa~\cite{prea2026tree} for recognizing strong \Trob\ spaces can be adapted to obtain a faster recognition algorithm for the general case.

A second natural direction concerns the optimization version of the problem. Given a dissimilarity space, one may seek a spanning tree that maximizes the number of Robinson paths. The computational complexity of this problem is currently unknown, and it remains to determine whether it is NP-hard or admits a polynomial-time algorithm. More generally, establishing theoretical approximation guarantees or developing exact algorithms for this optimization problem would provide stronger performance guarantees for practical applications.

Another direction is extending the framework to incomplete dissimilarity spaces, where some pairwise dissimilarities are missing. Such an extension is relevant in view of the increasing availability of large-scale relational datasets, especially in bioinformatics, where missing data naturally arise. Finally, it would be interesting to investigate further the applicability of \Trob spaces to other types of real-world data and to assess their usefulness in a broader range of data analysis. 




\end{document}